\documentclass[12pt,english,onecolumn, draftcls]{IEEEtran}
\usepackage[T1]{fontenc}
\usepackage[latin9]{inputenc}
\usepackage{amsthm}
\usepackage{amsmath}
\usepackage{amssymb}
\usepackage{esint}

\makeatletter
  \theoremstyle{definition}
  \newtheorem{defn}{Definition}
  \theoremstyle{plain}
  \newtheorem{lem}{Lemma}
  \theoremstyle{plain}
  \newtheorem{cor}{Corollary}
\theoremstyle{plain}
\newtheorem{thm}{Theorem}

\usepackage{cite}

\makeatother

\usepackage{babel}

\begin{document}

\title{Stochastic Ordering based Carrier-to-Interference Ratio Analysis
for the Shotgun Cellular Systems}

\author{Prasanna Madhusudhanan, Juan G. Restrepo, Youjian (Eugene) Liu, Timothy X Brown, Kenneth R. Baker \thanks{P. Madhusudhanan, Y. Liu, and T. X. Brown are with the Department of Electrical, Computer and Energy Engineering Department; J. G. Restrepo is with the Department of Applied Mathematics; and T. X. Brown and K. R. Baker are with the Interdisciplinary Telecommunications Program, at the University of Colorado, Boulder, CO 80309-0425 USA. Email: \{mprasanna, juanga, eugeneliu, timxb, kenneth.baker\}@colorado.edu}\vspace{-0.4in}
}
\maketitle
\begin{abstract}
A simple analytical tool based on stochastic ordering is developed
to compare the distributions of carrier-to-interference ratio at the
mobile station of two  cellular systems where the base stations are
distributed randomly according to certain non-homogeneous Poisson
point processes. The comparison is conveniently done by studying only
the base station densities without having to solve for the distributions
of the carrier-to-interference ratio, that are often hard to obtain.

\emph{Index Terms:} Carrier-to-interference Ratio, Co-channel Interference,
Fading channels, Stochastic ordering.
\end{abstract}

\section{Introduction\label{sec:Introduction}}

A Poisson point process has been adopted in the literature for the
locations of nodes in the study of shotgun cellular networks, ad-hoc
networks, and other uncoordinated and decentralized communication
networks \cite[and references therein]{Brown2000,Madhusudhanan2010a}.
An underlying assumption in all the previous work is that the density
of transmitters, referred to as base station (BS) throughout this
paper, is constant, i.e. the Poisson point process is homogeneous.
Such a model does not sufficiently represent reality.

In \cite{Madhusudhanan2010a} and here, we have modeled BS arrangement
by non-homogeneous Poisson point processes in $\mathbb{R}^{l},\ l=1,2,\ \mathrm{and}\ 3.$
We aim to characterize carrier-to-interference ratio $\left(\frac{C}{I}\right)$
at a given mobile station (MS). In \cite{Madhusudhanan2010a}, we
have derived semi-analytical expressions for the tail probability
of $\frac{C}{I},$ denoted as $\mathbb{P}\left(\left\{ \frac{C}{I}>y\right\} \right),$
by deriving the characteristic function of the reciprocal of $\frac{C}{I}.$
Moreover, the $\frac{C}{I}$ characterization holds for a wide range
of scenarios of interest such as arbitrary distributions for fading
and random BS transmission powers, arbitrary path-loss models and
arbitrary locations for the MS. In spite of such a general result,
it is still not convenient to compare the $\frac{C}{I}$ distributions
of two different networks.

Is it possible to qualitatively compare two $\frac{C}{I}$ distributions
by only examining the BS densities without having to obtain the $\frac{C}{I}$
distributions? This paper answers the question affirmatively for certain
BS densities by developing a stochastic ordering based tool. Concepts
of stochastic ordering have been applied to scenarios of interest
in wireless communications in \cite{Tepedelenlioglu2011}. To the
best of our knowledge, this is the first work that uses stochastic
ordering to understand large scale random wireless networks. The main
result of this paper is Theorem \ref{thm:st-order} for which Section
\ref{sub:ctoiso} develops the necessary tools. The utility of this
result is explored in Section \ref{sec:applicationd}, by considering
several scenarios of interest in modeling the wireless network.

\section{System Model\label{sec:modelreview}}

The \textit{Shotgun Cellular System (SCS) }is a model for the cellular
system in which the BSs are distributed in a $l$-dimensional plane
($l-$D, typically $l=1,2,\ \mathrm{and}\ 3$) according to a non-homogeneous
Poisson point process in $\mathbb{R}^{l}$. The intensity function
of the Poisson point process is called the BS density function.

Without loss of generality, we restrict our attention to 1-D SCSs,
because for the $\frac{C}{I}$ analysis, the $l$-D SCSs can be reduced
to an equivalent 1-D SCS \cite[Lemma 2]{Madhusudhanan2010a} with
a BS density function $\lambda(r),$ where $r\ge0$ is the distance
of the BS from a mobile-station (MS) located at the origin. For example,
a \emph{homogeneous} $l$-D SCS with density $\lambda_{0}\left(>0\right)$
is equivalent to a 1-D SCS with density function $\lambda\left(r\right)=\lambda_{0}b_{l}r^{l-1},\ r\ge0,\ b_{1}=2,\ b_{2}=2\pi,\ \mathrm{and}\ b_{3}=4\pi$
\cite[Corollary 1]{Madhusudhanan2010a}.

The BSs are assumed to have independent and identically distributed
(i.i.d.) random transmission powers $K_{i}$'s and shadow fadings
$\Psi_{i}$'s across BSs. The deterministic path-loss is $R^{-\varepsilon},\ \varepsilon>0$.
We assume an interference limited system and omit background noise.
We focus on the signal quality of a MS at the origin. The MS chooses
to communicate with the BS that corresponds to the strongest received
signal power, referred to as the {}``serving BS''. All other BSs
are the {}``interfering BSs''. The signal quality at the MS is measured
by $\frac{C}{I}=\frac{K_{S}\Psi_{S}R_{S}^{-\varepsilon}}{\sum_{i=1}^{\infty}K_{i}\Psi_{i}R_{i}^{-\varepsilon}},$
where $S$ indexes the serving BS and $i$ indexes the interfering
BSs. Further, $R_{S}\le R_{1}\le R_{2}\le\cdots$ are ordered BS locations.

\section{The Stochastic Ordering of $\frac{C}{I}$\label{sub:ctoiso}}

In this section, we present the theoretical tools that are used to
compare $\frac{C}{I}$ tail probability by comparing the equivalent
1-D BS densities $\lambda(r).$ Since the effect of i.i.d. shadow
fading factors and i.i.d. transmission powers can be captured by modifying
the BS density as shown in Section \ref{sub:Generalization-to-Networks},
they are assumed to be 1 for all BSs. The generalization to arbitrary
path loss model is given in \cite[Section VI]{Madhusudhanan2010a},
which is also equivalent to modifying $\lambda(r).$ As a result,
$\frac{C}{I}=\frac{R_{1}^{-\varepsilon}}{\sum_{i=2}^{\infty}R_{i}^{-\varepsilon}}$.
\begin{defn}
\label{def:UsualStOrder}Let $X$ and $Y$ be two random variables
such that $\mathbb{P}\left(\left\{ X>x\right\} \right)\le\mathbb{P}\left(\left\{ Y>x\right\} \right),\ \forall\ x\in\left(-\infty,\infty\right)$,
then X is \textit{smaller than} Y \textit{in the usual stochastic
order} and this is denoted by $X\le_{\mathrm{st}}Y$. Further, $X=_{\mathrm{st}}Y$
means $\mathbb{P}\left(\left\{ X>x\right\} \right)=\mathbb{P}\left(\left\{ Y>x\right\} \right),\ \forall\ x\in\left(-\infty,\infty\right).$
\cite[p. 3]{ShakedShanthikumar06}

If $X$ and $Y$ are the $\frac{C}{I}$ at the MS in two different
SCSs, $X\le_{\mathrm{st}}Y$ implies that the MS in the SCS corresponding
to $Y$ is more likely to achieve better signal quality than in the
SCS corresponding to $X$.
\end{defn}
Let $\left\{ R_{k}\right\} _{k=1}^{\infty}$ represent the set of
distances of BSs from the MS (indexed in the ascending order of the
distance), $D_{k+1}=R_{k+1}-R_{k}$ be the distance between two adjacent
BSs, and $f_{D_{k+1}|R_{k}}\left(\left.d\right|r;\lambda(s)\right)$
be the probability density function (p.d.f.) of $D_{k+1}$ conditioned
on $R_{k}=r$, as a function of the BS density $\lambda(s).$
\begin{lem}
\label{lem:scale}\emph{ \begin{eqnarray}
f_{D_{k+1}|R_{k}}\left(d\left|r;\lambda(s)\right.\right) & = & e^{-\int_{r}^{r+d}\lambda(s)ds}\lambda(r+d),\ \mathrm{and}\label{eq:pdf-D}\\
f_{aD_{k+1}|aR_{k}}\left(d'\left|r';\lambda(s)\right.\right) & = & f_{D_{k+1}|R_{k}}\left(d'\left|r';\frac{1}{a}\lambda\left(\frac{s}{a}\right)\right.\right).\label{eq:pdf-aD}\end{eqnarray}
}\end{lem}
\begin{proof}
Equation (\ref{eq:pdf-D}) follows from the properties of Poisson
processes \cite{Kingman1993,Ross1983}. Equation (\textbf{\ref{eq:pdf-aD}})
is proved by $f_{aD_{k+1}|aR_{k}}\left(d'|r';\lambda(s)\right)$ $\overset{\left(a\right)}{=}\frac{1}{a}f_{D_{k+1}|R_{k}}\left(\frac{d'}{a}|\frac{r'}{a};\lambda(s)\right)$
$\overset{\left(b\right)}{=}\frac{1}{a}\lambda\left(\frac{r'+d'}{a}\right)\exp\left(-\int_{\frac{r'}{a}}^{\frac{r'+d'}{a}}\lambda(s)ds\right)$
$\overset{\left(c\right)}{=}\frac{1}{a}\lambda\left(\frac{r'+d'}{a}\right)\exp\left(-\int_{r'}^{r'+d'}\frac{1}{a}\lambda\left(\frac{s'}{a}\right)ds'\right),$
where $\left(a\right)$ is obtained by a variable change; $\left(b\right)$
follows from (\ref{eq:pdf-D}); and $\left(c\right)$ is obtained
by a variable change and gives (\ref{eq:pdf-aD}).
\end{proof}
Lemma \ref{lem:scale} means that scaling $D_{k+1}$ and $R_{k}$
by $a$ is equivalent to scaling the BS density as $\frac{1}{a}\lambda(\frac{r}{a})$.
The significance of Lemma \ref{lem:scale} in the context of $\frac{C}{I}$
is as follows.
\begin{cor}
\label{cor:SIR-scaling}\emph{The distribution of $\frac{C}{I}$ at
the MS in the 1-D SCS with BS density function $\lambda(r)$ is the
same as that in 1-D SCSs with BS density functions $\frac{1}{a}\lambda(\frac{r}{a})$,
$\forall\ a>0$, i.e., $\left.\frac{C}{I}\right|_{\lambda(r)}=_{\text{st}}\left.\frac{C}{I}\right|_{\frac{1}{a}\lambda(\frac{r}{a})}$. }\end{cor}
\begin{proof}
Let $\left\{ R_{k}\right\} _{k=1}^{\infty}$ correspond to the 1-D
SCS with BS density function $\lambda(r).$ Then, since the ordered
BS locations $R_{k}$'s are determined by inter-BS distances, it follows
from Lemma \ref{lem:scale} that $\begin{array}{c}
\left.\frac{C}{I}\right|_{\lambda(r)}\end{array}=\left.\frac{(aR_{1})^{-\varepsilon}}{\sum_{k=2}^{\infty}(aR_{k})^{-\varepsilon}}\right|_{\lambda(r)}=_{\text{st}}\left.\frac{\left(R_{1}^{'}\right)^{-\varepsilon}}{\sum_{k=2}^{\infty}\left(R_{k}^{'}\right)^{-\varepsilon}}\right|_{\frac{1}{a}\lambda(\frac{r}{a})},$ where $R_{k}^{'}$'s corresponding to $\frac{1}{a}\lambda(\frac{r}{a})$
have the same distribution as $aR_{k}$'s with $\lambda(r)$.
\end{proof}
As a result, $\left\{ \frac{1}{a}\lambda\left(\frac{r}{a}\right),\ r\ge0,\ a>0\right\} $
forms a parametric family of BS density functions such that the 1-D
SCSs corresponding to them have the same $\frac{C}{I}$ at the MS.
In other words, appropriately scaling the BS density function will
not change the p.d.f. of $\frac{C}{I}$. Moreover, the following special
case is a direct corollary of the above result.
\begin{cor}
\emph{\label{cor:homogeousCTOIinvariance}In a }homogeneous\emph{
$l$-D SCS, $\frac{C}{I}$ is not a function of the BS density.}\end{cor}
\begin{proof}
Firstly, recall that the $\frac{C}{I}$ at the MS in a \emph{homogeneous}
$l-$D SCS with BS density $\lambda_{0}$ is the same as that in a
$1-$D SCS with a BS density function $\lambda\left(r\right)=\lambda_{0}b_{l}r^{l-1}.$
Next, from Corollary \ref{cor:SIR-scaling}, the distribution of $\frac{C}{I}$
in this SCS is the same as that in a 1-D SCS with the BS density function
$\frac{1}{a}\lambda\left(\frac{r}{a}\right)=\lambda_{0}\alpha b_{l}r^{l-1},\ \alpha=a^{-l},\ a>0.$
Thus, distributions of $\frac{C}{I}$ corresponding to $\alpha\lambda_{0}$
and $\lambda_{0}$ are the same.
\end{proof}
This was also  observed in \cite{Madhusudhanan2010a}, where we showed
that the expression for the characteristic function of $\left(\frac{C}{I}\right)^{-1}$
did not involve $\lambda_{0}.$ Corollary \ref{cor:homogeousCTOIinvariance}
provides a simpler and more fundamental proof. Next, we define a notation
used in Theorem \ref{thm:st-order}.
\begin{defn}
\label{def:DensityFunction}For BS density function $\lambda\left(r\right)$,
the cumulative BS density function is defined as $\mu(r)\triangleq\int_{0}^{r}\lambda(s)ds$,
and its inverse function is define as $\mu^{-1}(q)\triangleq\sup\{r:\mu(r)\le q\}$.
\end{defn}
Since $\lambda(r)\geq0$, $\mu\left(r\right)$ is a monotonically
increasing function of $r$. In general, the inverse function is not
injective since $\lambda(r)$ can be zero in arbitrary intervals of
$r\in\left[0,\infty\right).$ The above definition makes it injective.
For certain BS densities, it is possible to compare two $\frac{C}{I}$'s
by comparing the densities without solving for the distributions.
It is facilitated by Theorem \ref{thm:st-order}.
\begin{thm}
\label{thm:st-order}\emph{Let $\left\{ \lambda_{1}\left(r\right),\mu_{1}\left(r\right),\mu_{1}^{-1}\left(q\right)\right\} $
and $\left\{ \lambda_{2}\left(r\right),\mu_{2}\left(r\right),\mu_{2}^{-1}\left(q\right)\right\} $
be the BS density functions, cumulative BS density functions and their
inverse functions for two 1-D SCSs, respectively. The $\frac{C}{I}$
at the MS follows the stochastic order $\left.\frac{C}{I}\right|_{\lambda_{1}(r)}\le_{\text{st}}\left.\frac{C}{I}\right|_{\lambda_{2}(r)}$,
if for each $q>0$ and $a=\frac{\mu_{2}^{-1}(q)}{\mu_{1}^{-1}(q)}$,
$\frac{1}{a}\lambda_{1}\left(\frac{r}{a}\right)\ge\lambda_{2}(r),\;\forall\ r\ge\mu_{2}^{-1}(q)$}.

\end{thm}
See Appendix \ref{sec:st-order-proof} for the proof. Applications
of the above theorem are\emph{ }in the next section.

\section{\label{sec:applicationd}Applications of the $\frac{C}{I}$ stochastic
ordering}

\subsection{Comparison of Homogeneous $l-$D SCSs $\left(l=1,2,\ \mathrm{and}\ 3\right)$}

Here, we show that the signal quality degrades as the dimension $l$
of the \emph{homogeneous} $l$-D SCS increases, for which we need
the following corollaries.
\begin{cor}
\label{cor:monotone-diff}\emph{For each $q>0$ and $a=\frac{\mu_{2}^{-1}(q)}{\mu_{1}^{-1}(q)}$,
if $\Delta\lambda(r)\triangleq\frac{1}{a}\lambda_{1}(\frac{r}{a})-\lambda_{2}(r)$
is a non-decreasing function for all $r\ge0$, then $\left.\frac{C}{I}\right|_{\lambda_{1}(r)}\le_{\text{st}}\left.\frac{C}{I}\right|_{\lambda_{2}(r)}$.}\end{cor}
\begin{proof}
Note that $\underset{0}{\overset{\mu_{2}^{-1}\left(q\right)}{\int}}\frac{1}{a}\lambda_{1}\left(\frac{s}{a}\right)ds=q=\underset{0}{\overset{\mu_{2}^{-1}\left(q\right)}{\int}}\lambda_{2}(s)ds.$
Hence, $\underset{0}{\overset{\mu_{2}^{-1}\left(q\right)}{\int}}\Delta\lambda(s)ds=0$.
Suppose $\Delta\lambda(\mu_{2}^{-1}(q))<0$, then $\Delta\lambda(r)<0,\; r\in[0,\mu_{2}^{-1}(q)]$,
since $\Delta\lambda(r)$ is non-decreasing. This is a contradiction.
Thus, $\Delta\lambda(\mu_{2}^{-1}(q))\ge0.$ Using Theorem \ref{thm:st-order},
the corollary is proved.\end{proof}
\begin{cor}
\label{cor:homogenousStOrder}\emph{For a }homogeneous\emph{ $l$-D
SCS with BS density $\lambda_{0}$ and its equivalent 1-D BS density
function $\lambda_{l}\left(r\right)=\lambda_{0}b_{l}r^{l-1},\ r\ge0$,
multiplying $\lambda_{l}\left(r\right)$ with a non-increasing function
$\beta(r)>0$ improves the $\frac{C}{I}$, i.e., $\left.\frac{C}{I}\right|_{\lambda_{l}(r)}\le_{\text{st}}\left.\frac{C}{I}\right|_{\beta(r)\lambda_{l}(r)}$.
The inequality reverses if $\beta(r)$ is non-decreasing. }\end{cor}
\begin{proof}
If $\beta(r)$ is non-increasing, for any $a>0$, the density difference
$\Delta\lambda(r)=\frac{1}{a}\lambda_{l}(\frac{r}{a})-\beta(r)\lambda_{l}(r)=\lambda_{0}b_{l}\left(\frac{1}{a^{l-2}}-\beta(r)\right)r^{l-1}$
is non-decreasing. By Corollary \ref{cor:monotone-diff}, $\left.\frac{C}{I}\right|_{\lambda_{l}(r)}\le_{\text{st}}\left.\frac{C}{I}\right|_{\beta(r)\lambda_{l}(r)}$
holds. If $\beta(r)$ is non-decreasing, the same proof applies with
$\Delta\lambda(r)=\beta(r)\lambda_{l}(r)-a\lambda_{l}(\frac{r}{a})$.
\end{proof}

Hence, $\left.\frac{C}{I}\right|_{\lambda_{1}(r)}\overset{\left(a\right)}{\ge_{\text{st}}}\left.\frac{C}{I}\right|_{\lambda_{2}\left(r\right)}\overset{\left(b\right)}{\ge_{\text{st}}}\left.\frac{C}{I}\right|_{\lambda_{3}\left(r\right)}$
by plugging $l=1,2,3$ in $\lambda_{l}\left(r\right),$ respectively;
$\left(a\right)$ holds because $\lambda_{2}\left(r\right)=\beta\left(r\right)\lambda_{1}\left(r\right)$,
where $\beta\left(r\right)=\frac{b_{2}}{b_{1}}r$ is a non-decreasing
function; and similarly $\left(b\right)$ also holds. Thus, the comparison
between $\frac{C}{I}'s$ is done without finding their distributions.

\subsection{A Qualitative Comparison between Two 1-D SCS}

Consider a \emph{homogeneous} 1-D SCS with a BS density function $\lambda_{1}\left(r\right)=\lambda,\ r\ge0,$
and another 1-D SCS with a BS density function $\lambda_{2}\left(r\right)=\begin{cases}
\alpha & 0\le r\le\rho\\
\beta & r>\rho\end{cases},$ where $\alpha>\beta.$ Such $\lambda_{2}(r)$ might describe, for
example, a highway passing through a region of greater population
(BS density of $\alpha$) and then a region of smaller population
(BS density of $\beta$). Usually, such a scenario is approximated
by a constant BS density through out the highway, which is represented
by $\lambda_{1}\left(r\right).$ If $\lambda=\alpha$, it is easy
to guess that $\left.\frac{C}{I}\right|_{\lambda_{1}(r)}\le_{\text{st}}\left.\frac{C}{I}\right|_{\lambda_{2}(r)}$.
But if $\lambda>\alpha$, it is not clear which SCS has better $\frac{C}{I}.$
Theorem \ref{thm:st-order} shows that the SCS with density $\lambda_{2}(r)$
has better $\frac{C}{I}$ than $\lambda_{1}\left(r\right),$ and the
result holds irrespective of the specific values of $\alpha,\ \beta,$
and $\lambda.$ To apply Theorem \ref{thm:st-order}, we note that
$\mu_{1}^{-1}(q)=\frac{q}{\lambda},\ q\ge0,\ \mathrm{and}\ \mu_{2}^{-1}(q)=\begin{cases}
\frac{q}{\alpha} & ,\ q\leq\alpha\rho\\
\frac{q+(\beta-\alpha)\rho}{\beta} & ,\ q>\alpha\rho\end{cases}.$ Further, $a(q)=\begin{cases}
\frac{\lambda}{\alpha} & ,\ q\le\alpha\rho,\\
\frac{1+(\beta-\alpha)\rho/q}{\beta/\lambda} & ,\ q>\alpha\rho\end{cases}.$ As a result, for $q\leq\alpha\rho,$ $\frac{1}{a}\lambda_{1}(\frac{r}{a})=\begin{cases}
\alpha\geq\alpha=\lambda_{2}(r) & ,\ \frac{q}{\alpha}<r<\rho\\
\alpha\geq\beta=\lambda_{2}(r) & ,\ r>\rho\end{cases},$ and for $q>\alpha\rho,$ $\mu_{2}^{-1}(q)>\rho,$ $\frac{1}{a}\lambda_{1}(\frac{r}{a})=\frac{\beta}{1+(\beta-\alpha)\rho/q}\geq\beta=\lambda_{2}(r),\ r>\rho.$
Thus, applying Theorem \ref{thm:st-order},$\left.\frac{C}{I}\right|_{\lambda_{1}(r)}\le_{\text{st}}\left.\frac{C}{I}\right|_{\lambda_{2}(r)}$.
Similarly, if $\alpha<\beta,$ we can show $\left.\frac{C}{I}\right|_{\lambda_{1}(r)}\ge_{\text{st}}\left.\frac{C}{I}\right|_{\lambda_{2}(r)}.$

\subsection{\label{sub:comparePL}Comparison of Path-loss Models}

Here, we compare the $\frac{C}{I}$ at the MS in two \emph{homogeneous}
$l-$D SCSs with a BS density $\lambda_{0}$ and with different path-loss
models, $\frac{1}{h_{1}\left(r\right)},$ and $\frac{1}{h_{2}\left(r\right)},\ \forall\ r\ge0$.
The proofs for Corollary \ref{cor:cmp_exponents} and \ref{cor:cmp_var_exponents}
involve the following common steps. Firstly, using \cite[Lemma 2]{Madhusudhanan2010a},
reduce the homogeneous $l-$D SCS to the equivalent 1-D SCS with a
BS density function $\lambda_{i}\left(r\right)=\lambda_{0}b_{l}r^{l-1},\ i=1,\ 2.$
Next, using \cite[Theorem 5]{Madhusudhanan2010a}, the resultant 1-D
SCSs can be reduced to an equivalent 1-D SCS with path-loss exponent
$\varepsilon=1$, i.e., $\frac{1}{h_{i}\left(r\right)}=\frac{1}{r}$,
and with BS density functions $\bar{\lambda}_{i}\left(r\right)=\frac{\lambda_{i}\left(h_{i}^{-1}\left(r\right)\right)}{h_{i}^{'}\left(h_{i}^{-1}\left(r\right)\right)},$
$i=1,2$ where $h_{i}^{-1}\left(\cdot\right)$ is the inverse function
of $h_{i}\left(\cdot\right).$ Finally, we have two 1-D SCSs with
the same path-loss model $\frac{1}{h_{i}\left(r\right)}=\frac{1}{r}$
and BS density functions $\bar{\lambda}_{i}\left(r\right),\ r\ge0,\ i=1,2.$
We show that Theorem \ref{thm:st-order} applies and derive the result.
\begin{cor}
\emph{\label{cor:cmp_exponents}In a }homogeneous\emph{ $l-$D SCS,
if the path-loss follows a power-law parametrized by a path-loss exponent,
$\varepsilon$, the $\frac{C}{I}$ at the MS improves as the path-loss
exponent increases. In other words, if $h_{i}\left(r\right)=r^{\varepsilon_{i}},\ i=1,2,$
such that $\varepsilon_{1}>\varepsilon_{2}>l,$ then $\left(\frac{C}{I}\right)_{1}\ge_{\mathrm{st}}\left(\frac{C}{I}\right)_{2},$
where $\left(\frac{C}{I}\right)_{i}$ corresponds to the path-loss
model $\frac{1}{h_{i}\left(r\right)}.$}\end{cor}
\begin{proof}
At the end of Step 2, the equivalent 1-D SCSs with a path-loss model
$\frac{1}{r}$ have the BS density functions $\bar{\lambda}_{i}\left(r\right)=\frac{\lambda_{0}b_{l}}{\varepsilon_{i}}r^{\frac{l}{\varepsilon_{i}}-1}.$
Further, $\bar{\lambda}_{2}\left(r\right)=\beta\left(r\right)\bar{\lambda}_{1}\left(r\right),$
where $\beta\left(r\right)=\frac{\varepsilon_{1}}{\varepsilon_{2}}r^{\frac{l}{\varepsilon_{2}}-\frac{l}{\varepsilon_{1}}},\ r\ge0$
is a non-decreasing function. Hence, Corollary \ref{cor:homogenousStOrder}
applies and $\left.\frac{C}{I}\right|_{\bar{\lambda}_{1}\left(r\right)}\ge_{\mathrm{st}}\left.\frac{C}{I}\right|_{\bar{\lambda}_{2}\left(r\right)}.$
\end{proof}

Hence, a simple proof that does not require solving the distribution
of $\frac{C}{I}$ gives the expected result that a channel with a
greater path-loss exponent has a better $\frac{C}{I}$. The following
corollary establishes a similar result between two popularly used
path-loss models \cite{Rappaport1996}.
\begin{cor}
\emph{\label{cor:cmp_var_exponents}In a }homogeneous\emph{ $l-$D
SCS with a BS density $\lambda_{0},$ the received signal of a MS
located at the origin satisfies $\left(\frac{C}{I}\right)_{1}\le_{\mathrm{st}}\left(\frac{C}{I}\right)_{2},$
where $\left(\frac{C}{I}\right)_{1}$ corresponds to path-loss $\frac{1}{h_{1}\left(r\right)}$
with $h_{1}\left(r\right)=r^{\varepsilon_{1}},\ r\ge0$ and $\left(\frac{C}{I}\right)_{2}$
corresponds to the path-loss $\frac{1}{h_{2}\left(r\right)}$ with
$h_{2}\left(r\right)=\begin{cases}
r^{\varepsilon_{1}} & ,\ r\le1\\
r^{\varepsilon_{2}} & ,\ r>1\end{cases},$ where $\varepsilon_{2}>\varepsilon_{1}>l.$ The opposite conclusion
holds when $\varepsilon_{1}>\varepsilon_{2}>l.$}\end{cor}
\begin{proof}
At the end Step 2, $\bar{\lambda}_{1}\left(r\right)=\frac{\lambda_{0}b_{l}}{\varepsilon_{1}}r^{\frac{l}{\varepsilon_{1}}-1},\ r\ge0,$
and $\bar{\lambda}_{2}\left(r\right)$ satisfies the equation $\bar{\lambda}_{2}\left(r\right)\beta\left(r\right)=\bar{\lambda}_{1}\left(r\right),$
where $\beta\left(r\right)=\begin{cases}
1 & ,\ r\le1\\
\frac{\varepsilon_{1}}{\varepsilon_{2}}r^{\frac{l}{\varepsilon_{2}}-\frac{l}{\varepsilon_{1}}} & ,\ r>1\end{cases}.$ Since $\varepsilon_{2}>\varepsilon_{1}>l,$ $\beta\left(r\right)$
is a non increasing function. As a result, Corollary \ref{cor:homogenousStOrder}
holds and hence $\left.\frac{C}{I}\right|_{\bar{\lambda}_{1}\left(r\right)}\le_{\mathrm{st}}\left.\frac{C}{I}\right|_{\bar{\lambda}_{2}\left(r\right)}$.
Thus, the system with the path-loss model $\frac{1}{h_{2}\left(r\right)}$
has a better signal quality compared to that of $\frac{1}{h_{1}\left(r\right)}.$
Now, when $\varepsilon_{1}>\varepsilon_{2}>l,$ $\beta\left(r\right)$
is a non decreasing function and $\left.\frac{C}{I}\right|_{\bar{\lambda}_{1}\left(r\right)}\ge_{\mathrm{st}}\left.\frac{C}{I}\right|_{\bar{\lambda}_{2}\left(r\right)}.$
\end{proof}

\subsection{\label{sub:Generalization-to-Networks}Shadow Fading and Random Transmission
Powers}

In all the results until now, the shadow fading factors and the transmission
powers for all the BSs were constant. Here, we generalize to the case
when they are random variables, i.i.d. across BSs. The transmission
power and the shadow fading factor of the same BS could be dependent.
The following result reduces the 1-D SCS with random shadow fading
factors and random transmission powers to a 1-D SCS where both are
1.
\begin{thm}
\label{cor:shadowing}\emph{For the $\frac{C}{I}$ analysis in a 1-D
SCS with BS density function $\lambda\left(r\right),$ if the random
shadow fading factors $\left\{ \Psi_{i}\right\} _{i=1}^{\infty}$
and transmission powers $\left\{ K_{i}\right\} _{i=1}^{\infty}$ are
i.i.d. across all the BSs, the SCS is equivalent to another 1-D SCS
with BS density function $\bar{\lambda}\left(r\right)=\mathbb{E}_{\Psi,K}\left[\left(\Psi K\right)^{\frac{1}{\varepsilon}}\lambda\left(r\left(\Psi K\right)^{\frac{1}{\varepsilon}}\right)\right],$
where $\mathbb{E}$ is the expectation operator w.r.t. $\Psi$ and
$K,$ which has the same distribution as $\Psi_{i}$ and $K_{i},\ \forall\ i,$
respectively. This holds as long as the expectation converges.}\end{thm}
\begin{proof}
For the random shadow fading and transmission powers case, the $\frac{C}{I}$
defined in Section \ref{sec:modelreview} can be written as $\frac{C}{I}=\frac{\left(R_{j}\left(K_{j}\Psi_{j}\right)^{-\frac{1}{\varepsilon}}\right)^{-\varepsilon}}{\underset{i,\ i\ne j}{\sum}\left(R_{i}\left(K_{i}\Psi_{i}\right)^{-\frac{1}{\varepsilon}}\right)^{-\varepsilon}},$
where the index $j$ corresponds to the BS with the strongest received
signal power at the MS. The above expression for the $\frac{C}{I}$
corresponds to a 1-D SCS with distances from the MS given by $\left\{ R_{i}\left(K_{i}\Psi_{i}\right)^{-\frac{1}{\varepsilon}}\right\} _{i=1}^{\infty},$
unity shadow fading factors and unity transmission powers at each
BS. Now, \cite[Theorem 4]{Madhusudhanan2010a} applies, and the new
1-D SCS has a BS density function $\bar{\lambda}\left(r\right).$
\end{proof}
Thus, if there are random shadow fading factors and/or transmission
powers at each BS, one can first apply Theorem \ref{cor:shadowing}
to obtain the equivalent 1-D SCSs with constant shadowing and transmission
powers, and then apply Theorem \ref{thm:st-order} for the comparison
of $\frac{C}{I}$'s. The following corollary shows a scenario where
$\frac{C}{I}$ distribution is unaffected by the distributions of
random shadow fading factors and transmission powers.
\begin{cor}
\label{cor:shadowingInvariance}\emph{In a }homogeneous\emph{ $l-$D
SCS with a BS density $\lambda_{0}$, the $\frac{C}{I}$ distribution
at the MS does not depend on the distributions of the random shadow
fading factors $\left\{ \Psi_{i}\right\} _{i=1}^{\infty}$ and transmission
powers $\left\{ K_{i}\right\} _{i=1}^{\infty},$ if they are i.i.d.
across BSs and }\textup{$\left|\mathbb{E}_{\Psi,K}\left[\left(\Psi K\right)^{\frac{l}{\varepsilon}}\right]\right|<\infty$}\emph{.}\end{cor}
\begin{proof}
Firstly, recall that the homogeneous $l-$D SCS is equivalent to a
1-D SCS with a BS density function $\lambda\left(r\right)=\lambda_{0}b_{l}r^{l-1},\ r\ge0.$
We have $\left(\frac{C}{I}\right)_{\text{rand}}\overset{\left(a\right)}{=}_{\mathrm{st}}\left.\frac{C}{I}\right|_{\bar{\lambda}\left(r\right)}\overset{\left(b\right)}{=}_{\mathrm{st}}\left.\frac{C}{I}\right|_{\frac{1}{\alpha}\lambda\left(\frac{r}{\alpha}\right)}\overset{\left(c\right)}{=}_{\mathrm{st}}\left.\frac{C}{I}\right|_{\lambda\left(r\right)}$,
where $\left(a\right)$ is obtained by applying Theorem \ref{cor:shadowing}
to the 1-D SCS with the BS density function $\lambda\left(r\right),$
to obtain the equivalent 1-D SCS with constant shadow fading factors
and transmission powers and with a BS density $\bar{\lambda}\left(r\right)=\mathbb{E}_{\Psi,K}\left[\left(\Psi K\right)^{\frac{1}{\varepsilon}}\lambda\left(r\left(\Psi K\right)^{\frac{1}{\varepsilon}}\right)\right]=\mathbb{E}_{\Psi,K}\left[\left(\Psi K\right)^{\frac{l}{\varepsilon}}\right]\lambda\left(r\right)$;
$\left(b\right)$ is obtained by rewriting $\bar{\lambda}\left(r\right)$
as $\frac{1}{\alpha}\lambda\left(\frac{r}{\alpha}\right)$ where $\alpha=\left(\mathbb{E}_{\Psi,K}\left[\left(\Psi K\right)^{\frac{l}{\varepsilon}}\right]\right)^{-\frac{1}{l}}$;
$\left(c\right)$ is obtained by applying Corollary \ref{cor:SIR-scaling}.
Thus, the shadow fading and random transmission powers have no effect
on the $\frac{C}{I}$ distribution.
\end{proof}
This result was already proved in \cite[Remark 4(a)]{Madhusudhanan2010a}.
But here, we have shown an elegant alternative proof that is based
only on the concepts of stochastic ordering. Finally, as a consequence
of Corollary \ref{cor:shadowingInvariance}, the results in Section
\ref{sub:comparePL} also hold for cases with random shadow fading
factors and transmission powers that are i.i.d. across BSs.

\section{\label{sec:Conclusions}Conclusions}

This paper is an extension to our previous work in characterizing
the $\frac{C}{I}$ of a SCS in \cite{Madhusudhanan2010a}. The study
of the $\frac{C}{I}$ at the MS in a cellular system with BSs distributed
according to a \emph{non-homogeneous} Poisson process is difficult
because the distribution of the $\frac{C}{I}$ is not in closed form
\cite{Madhusudhanan2010a} and it is difficult to form an intuition
about such networks. As a result, most of the $\frac{C}{I}$ analysis
are restricted to the \emph{homogeneous} $l-D$ SCSs. Here, we have
developed a stochastic ordering based tool to analyze the $\frac{C}{I}$
at the MS in such non-homogeneous Poisson processes. Due to Theorem
\ref{thm:st-order}, for certain BS densities, we show that, by just
comparing the BS density functions of the SCSs, we can make strong
inferences such as, a MS in a given SCS achieves a $\frac{C}{I}$
that is at least as good as that achieved in another SCS without having
to solve for the $\frac{C}{I}$ distributions. Moreover, as a consequence
of Theorem \ref{thm:st-order}, elegant proofs are derived to show
that (1) a MS in a \emph{homogeneous} $l$-D SCS sees decreasing signal
quality as dimension $l$ increases; (2) the $\frac{C}{I}$ at the
MS improves as the path-loss exponent of the channel increases; and
(3) as far as $\frac{C}{I}$ is concerned, a SCS with random shadow
fading factors and random transmission powers, which are i.i.d. across
BSs, is equivalent to a SCS with constant shadow fading factors and
transmission powers and with a modified BS density.

\appendix

\subsection{Proof of Theorem \ref{thm:st-order}\label{sec:st-order-proof}}

Consider the 1-D SCS specified by the set $\left\{ \lambda\left(r\right),\mu\left(r\right),\mu^{-1}\left(q\right)\right\} $,
as in Definition \ref{def:DensityFunction}. The following remark
relates $\frac{C}{I}$ to the cumulative BS density.

If $R_{1}$ denotes the distance between the serving BS and MS in
the 1-D SCS, \[
\mathbb{P}\left(\left\{ \frac{C}{I}>y\right\} \right)\overset{\left(a\right)}{=}\int_{r_{1}=0}^{\infty}\mathbb{P}\left(\left.\frac{C}{I}>y\right|R_{1}=r\right)f_{R_{1}}(r)dr\overset{\left(b\right)}{=}\int_{q=0}^{\infty}\mathbb{P}\left(\left.\frac{C}{I}>y\right|Q=q\right)f_{Q}(q)dq,\]
where $Q\triangleq\mu(R_{1}),$ and $Q$ is an exponential random
variable with mean 1.

Equation $\left(a\right)$ is obtained by conditioning w.r.t. $R_{1}$.
Equation $\left(b\right)$ is obtained by expressing $\left(a\right)$
in terms of $Q$, where the p.d.f. of $R_{1}$ at $R_{1}=\mu^{-1}\left(q\right)$
is $\begin{array}{c}
\left.f_{R_{1}}(r)dr\right|_{r=\mu^{-1}(q)}\end{array}=\left.e^{-\int_{0}^{r}\lambda(s)ds}\lambda(r)dr\right|_{r=\mu^{-1}(q)}=e^{-q}dq=f_{Q}(q)dq,$ which does not depend on $\lambda(r)$.

To show that the BS density $\lambda_{1}(r)$ gives a worse $\frac{C}{I}$
than $\lambda_{2}(r)$ does, one needs to show that $\left.\frac{C}{I}\right|_{R_{1}=\mu_{1}^{-1}(q),\ \lambda_{1}(r),\ r\ge\mu_{1}^{-1}(q)}\le_{\text{st}}\left.\frac{C}{I}\right|_{R_{1}=\mu_{2}^{-1}(q),\ \lambda_{2}(r),\ r\ge\mu_{2}^{-1}(q)}$
for all $q>0$, where the condition of the domain of the BS density
is because the locations of interfering BSs only depend on the BS
density in that domain. Next, define $a=\frac{\mu_{2}^{-1}(q)}{\mu_{1}^{-1}(q)}$.
By Corollary \ref{cor:SIR-scaling},\begin{eqnarray*}
\left.\frac{C}{I}\right|_{R_{1}=\mu_{1}^{-1}(q),\ \lambda_{1}(r),\ r\ge\mu_{1}^{-1}(q)} & = & \left.\frac{\left(\mu_{1}^{-1}(q)\right)^{-\varepsilon}}{\sum_{k=2}^{\infty}R_{k}^{-\varepsilon}}\right|_{R_{2}\ge\mu_{1}^{-1}(q),\ \lambda_{1}(r),\ r\ge\mu_{1}^{-1}(q)}\end{eqnarray*}
\begin{eqnarray*}
=\left.\frac{\left(a\mu_{1}^{-1}(q)\right)^{-\varepsilon}}{\sum_{k=2}^{\infty}(aR_{k})^{-\varepsilon}}\right|_{R_{2}\ge\mu_{1}^{-1}(q),\ \lambda_{1}(r),\ r\ge\mu_{1}^{-1}(q)} & =_{\text{st}} & \left.\frac{\left(\mu_{2}^{-1}(q)\right)^{-\varepsilon}}{\sum_{k=2}^{\infty}\left(R_{k}^{'}\right)^{-\varepsilon}}\right|_{R_{2}^{'}\ge\mu_{2}^{-1}(q),\ \frac{1}{a}\lambda_{1}(\frac{r}{a}),\ r\ge\mu_{2}^{-1}(q)},\end{eqnarray*}
 where $R_{k}^{'}$'s are the ordered BS locations of the SCS with
BS density $\frac{1}{a}\lambda\left(\frac{r}{a}\right)$. The equation
means that the conditional $\frac{C}{I}$ of the SCS with a BS density
$\lambda_{1}(r)$ is equivalent to an SCS with BS density $\frac{1}{a}\lambda_{1}(\frac{r}{a})$
with the same location of the serving BS as the SCS with BS density
$\lambda_{2}(r)$.

With the locations of the serving BSs equal and fixed, $\frac{C}{I}$
is a decreasing function of the interference. Theorem 1.A.3.(a) of
\cite{ShakedShanthikumar06} says that decreasing functions reverse
stochastic order. Therefore, one only needs to show that the interferences
satisfy \begin{eqnarray}
\sum_{k=2}^{\infty}R_{k}^{-\varepsilon}|_{R_{2}\ge\mu_{2}^{-1}(q),\ \frac{1}{a}\lambda_{1}(\frac{r}{a}),\ r\ge\mu_{2}^{-1}(q)} & \ge_{\text{st}} & \sum_{k=2}^{\infty}R_{k}^{-\varepsilon}|_{R_{2}\ge\mu_{2}^{-1}(q),\ \lambda_{2}(r),\ r\ge\mu_{2}^{-1}(q)}.\label{eq:interference-order}\end{eqnarray}
 As shown in \cite[Appendix B]{Madhusudhanan2010a}, the total interference
power can be expressed as $\sum_{k=2}^{\infty}R_{k}^{-\varepsilon}=\lim_{r_{B}\rightarrow\infty}\lim_{N\rightarrow\infty}\sum_{i=2}^{N}X_{i}$,
where $X_{i}$ is a Bernoulli random variable defined by \[
\mathbb{P}\left(\left\{ \left.X_{i}=0\right|R_{1}=r_{1}\right\} \right)=p_{i},\ \mathbb{P}\left(\left\{ \left.X_{i}=r_{i}^{-\varepsilon}+o\left(\Delta r\right)\right|R_{1}=r_{1}\right\} \right)=1-p_{i},\]
 $p_{i}=\lambda(r_{i})\Delta r+o(\Delta r)$, $r_{i}=r_{1}+(i-1)\Delta r$,
$\Delta r=\frac{r_{B}-r_{1}}{N}$, and $r_{1}=\mu_{2}^{-1}(q)$. Now,
since the condition $\frac{1}{a}\lambda_{1}(\frac{r}{a})\ge\lambda_{2}(r)$
holds for all $r\ge\mu_{2}^{-1}(q)$, we have $X_{i}|_{\frac{1}{a}\lambda_{1}(\frac{r}{a})}\ge_{\text{st}}X_{i}|_{\lambda_{2}(r)},\;\forall i\ge2.$
Since summation preserves the stochastic order \cite[Theorem 1.A.3.(b)]{ShakedShanthikumar06},
(\ref{eq:interference-order}) is proved, completing the proof of
Theorem \ref{thm:st-order}.

\bibliographystyle{IEEEtran}
\bibliography{stochasticOrdering}

\begin{thebibliography}{1}
\providecommand{\url}[1]{#1}
\csname url@samestyle\endcsname
\providecommand{\newblock}{\relax}
\providecommand{\bibinfo}[2]{#2}
\providecommand{\BIBentrySTDinterwordspacing}{\spaceskip=0pt\relax}
\providecommand{\BIBentryALTinterwordstretchfactor}{4}
\providecommand{\BIBentryALTinterwordspacing}{\spaceskip=\fontdimen2\font plus
\BIBentryALTinterwordstretchfactor\fontdimen3\font minus
  \fontdimen4\font\relax}
\providecommand{\BIBforeignlanguage}[2]{{%
\expandafter\ifx\csname l@#1\endcsname\relax
\typeout{** WARNING: IEEEtran.bst: No hyphenation pattern has been}%
\typeout{** loaded for the language `#1'. Using the pattern for}%
\typeout{** the default language instead.}%
\else
\language=\csname l@#1\endcsname
\fi
#2}}
\providecommand{\BIBdecl}{\relax}
\BIBdecl

\bibitem{Brown2000}
T.~X. Brown, ``Cellular performance bounds via shotgun cellular systems,''
  \emph{IEEE Journal on Selected Areas in Communications}, vol.~18, no.~11, pp.
  2443--2455, Nov 2000.

\bibitem{Madhusudhanan2010a}
P.~Madhusudhanan, J.~G. Restrepo, Y.~Liu, T.~X. Brown, and K.~Baker,
  ``Generalized carrier to interference ratio analysis for the shotgun cellular
  system in multiple dimensions,'' \emph{CoRR}, vol. abs/1002.3943, 2010.

\bibitem{Tepedelenlioglu2011}
C.~Tepedelenlioglu, A.~Rajan, and Y.~Zhang, ``Applications of stochastic
  ordering to wireless communications,'' \emph{CoRR}, vol. abs/1101.4617, 2011.

\bibitem{ShakedShanthikumar06}
M.~Shaked and J.~G. Shanthikumar, \emph{Stochastic Orders}.\hskip 1em plus
  0.5em minus 0.4em\relax Springer, 2006.

\bibitem{Kingman1993}
J.~F.~C. Kingman, \emph{Poisson Processes (Oxford Studies in
  Probability)}.\hskip 1em plus 0.5em minus 0.4em\relax Oxford University
  Press, USA, January 1993.

\bibitem{Ross1983}
S.~M. Ross, \emph{Stochastic Processes}.\hskip 1em plus 0.5em minus 0.4em\relax
  John Wiley \& Sons, Inc., 1983.

\bibitem{Rappaport1996}
T.~S. Rappaport, \emph{Wireless Communications: Principles and Practice}.\hskip
  1em plus 0.5em minus 0.4em\relax Upper Saddle River, NJ, USA: Prentice-Hall,
  Inc., 1996.

\end{thebibliography}

\end{document}